\documentclass[a4paper,UKenglish,cleveref, autoref, thm-restate]{lipics-v2021}

\pdfoutput=1 
\hideLIPIcs

\usepackage{microtype}
\usepackage{upgreek}
\usepackage{amsmath, amssymb}
\usepackage{tcolorbox}
\usepackage[ruled,vlined]{algorithm2e}
\usepackage{amsmath}
\usepackage{tikz}
\usetikzlibrary{shapes.geometric}

\newcommand{\fO}{\mathcal{O}}

\title{Exact algorithms for optimal discretization}

\author{L\'aszl\'o Kozma}{Faculty of Computer Science, TU Dresden, Germany}{laszlo.kozma@tu-dresden.de}{}{}

\author{Junqi Tan}{Institut für Informatik, Freie Universität Berlin, Germany}{tanjunqi@zedat.fu-berlin.de}{}{}

\authorrunning{L. Kozma and J. Tan}

\Copyright{L\'{a}szl\'{o} Kozma and Junqi Tan} 

\relatedversion{}
\relatedversiondetails{Full Version}{https://arxiv.org/abs/2608.25197}

\ccsdesc[500]{Theory of computation~Design and analysis of algorithms} 

\keywords{discretization, point set, separation, exponential time} 

\acknowledgements{
We thank Günter Rote for suggesting an idea that
allowed us to simplify the algorithms presented
in our IPEC 2026 paper.
}

\nolinenumbers

\EventEditors{Tuukka Korhonen and Michael Lampis}
\EventNoEds{2}
\EventLongTitle{21th International Symposium on Parameterized and Exact Computation (IPEC 2026)}
\EventShortTitle{IPEC 2026}
\EventAcronym{IPEC}
\EventYear{2026}
\EventDate{September 2--4, 2026}
\EventLocation{L’Aquila, Italy}
\EventLogo{}
\SeriesVolume{398}
\ArticleNo{8}

\begin{document}
\nolinenumbers

\maketitle

\begin{abstract}

The \emph{optimal discretization problem} asks, given two disjoint sets of points $R$ and $B$ in the plane, for a minimal family of horizontal and vertical lines that separate the two sets, so that no cell delimited by the lines contains points from both sets. The problem arises as a pre-processing in supervised machine learning, and has received significant attention in parameterized algorithmics. Answering the question raised by Bonnet, Giannopoulos, and Lampis [IPEC 2017] and Froese [PhD thesis, 2018], it was shown by  Kratsch, Masa\v{r}\'{i}k, Muzi, Pilipczuk, and Sorge [SODA 2021] that optimal discretization admits a fixed-parameter algorithm with running time  $2^{\fO(k^2 \log{k})} \cdot n^{\fO(1)}$, where $k$ is the solution size and $n = |R| + |B|$.

In this paper we give an algorithm for optimal discretization that runs in time $\fO(1.9602^n)$. We also study the related \emph{point separation} problem that asks to separate \emph{all} input points by axis-parallel lines. For this problem we obtain an algorithm with runtime $\fO(1.8899^n)$. Our guarantees follow from structural observations about bichromatic and monochromatic point sets. To our knowledge, these are the first improvements over the trivial $2^n$ bound for both problems.

\end{abstract}

\section{Introduction}
\label{sec:Introduction}

We study two natural geometric separation problems for planar point sets. In the first problem, given two disjoint sets of points $R$ and $B$ in the plane, the task is to find a set of at most $k$ axis-parallel lines avoiding the input points, so that every cell induced by these lines is either empty or contains points from only one of the two sets. 

The problem is known as \emph{optimal discretization} and its original motivation comes from pre-processing in machine learning tasks. Given a set of data points with continuous numerical features and class labels, we would like to round the features to discrete values using a set of well-chosen thresholds, such as to guarantee that no two points with different class labels are mapped to the same point. Here the horizontal and vertical lines represent the thresholds for two features and the grid induced by these lines gives a discretized view of the data. (For a {broader} overview of the history of the problem in the context of machine learning, we refer to~\cite[\S\,5]{Froese18}, \cite{KratschMMPS21}, and references therein.)

Optimal discretization has been well studied both in exact and approximate settings. The problem was shown to be APX-hard by C\u{a}linescu, Dumitrescu, Karloff, and Wan~\cite{CalinescuDKW05}; see also~\cite{Megiddo88, Chlebus} for earlier NP-hardness proofs. 
Bonnet, Giannopoulos, and Lampis~\cite{BonnetGL19} showed that the problem is FPT when parameterized by the size of the smaller set $t = \min\{|R|,|B|\}$ and gave an $\fO^*(9^{t})$-time algorithm using a reduction to \textsc{2-SAT}. They also conjectured that the problem is FPT when parameterized by the solution size $k$, a question also raised by Froese~\cite[\S\,5.5]{Froese18}, as a special case of \emph{co-clustering}. The FPT status of the problem was confirmed by Kratsch, Masa\v{r}\'{i}k, Muzi, Pilipczuk, and Sorge~\cite{KratschMMPS21}. 
Their approach starts with a coarse solution obtained from an approximation, which is gradually transformed into an exact solution, by formulating the refinement as a structured CSP problem; the final running time is of the form $2^{\fO(k^2\log k)}\cdot n^{\fO(1)}$. 

The problem admits an almost trivial algorithm with running time $\fO^*(2^n)$ or $\fO^{*}(n^k)$, as follows. Let $H$ and $V$ be the horizontal and vertical lines of the solution. Guessing either $H$ or $V$, the other set of lines can be placed optimally using a simple greedy strategy. Since there are at most $n-1$ essentially different locations for horizontal and for vertical lines (between coordinates of the input points) the ``guessing'' amounts to trying up to ${n \choose k} \leq \min\{n^k, 2^n\}$ subsets of locations, while the second phase requires only polynomial time. The running time can be improved to $\fO^{*}(n^{k/2})$, if we guess the \emph{smaller} of the two sets $H$ and $V$, of size at most $k/2$. This however, would improve the $2^n$ bound only by a polynomial factor. (Note that if $k$ is unknown, it can be found, e.g., by binary search.) 

For moderately large values of the solution size, e.g., $k \in \Omega(\log{{n}})$, the FPT algorithm mentioned above is slower than the naive exponential algorithms, and as pointed out by Kratsch et al.~\cite{KratschMMPS21}, improving the dependence on $k$ in the FPT result appears to be difficult. This motivates designing algorithms with better dependence on $n$ in the exponential regime.

In this paper we show the following main result. 

\begin{theorem}\label{thm1}
The optimal discretization problem can be solved in time $\fO(1.9602^n)$.
\end{theorem}

A natural variant of the problem is to have a single input point set $P$, asking for a set of at most $k$ axis-parallel lines avoiding the points, so that every cell induced by the lines contains at most one input point. This problem can be seen either as a natural geometric separation task in itself, or as a multiclass version of the above discretization problem where each point has a distinct class label. We refer to this problem simply as \emph{point separation}. 

Similarly to optimal discretization, point separation is known to be NP-hard and APX-hard, and both problems admit a polynomial time $2$-approximation~\cite{CalinescuDKW05}. For point separation an FPT-algorithm easily follows from the observation that to create at least $n$ cells, the solution size must be at least $\sqrt{n}$. Exponential-time enumeration can be applied, similarly to optimal discretization, yielding a running time of ${n \choose k} \le {k^2 \choose k} \in 2^{\fO(k\log k)}$ (omitting polynomial factors). 

Our second result gives an improved exponential algorithm for point separation.

\begin{theorem}\label{thm2}
The point separation problem can be solved in time $\fO(1.8899^n)$.
\end{theorem}

To our knowledge, these are the first algorithms with runtime below the $2^n$-barrier for both problems. Several well-studied problems admit algorithms with runtime of the form $2^n$, typically based on subset enumeration, and designing algorithms with runtimes $c^n$ with $c<2$ has been a central theme in the field of exponential algorithms~\cite{FominK10}. Landmark results of this kind include algorithms for \emph{maximum independent set}~\cite{TarjanT77}, \emph{subset sum}~\cite{HorowitzS74}, \emph{Hamiltonian cycle}~\cite{Bjorklund14}, \emph{max cut}~\cite{Williams05}, and \emph{$3$-CNF SAT}~\cite{Schoning99}. (For each of the mentioned problems, $n$ is the standard problem-specific input parameter.)  
Prominent problems that have resisted such improvements include \emph{graph coloring}, \emph{CNF SAT}, \emph{set cover}, and \emph{TSP}.

The search for algorithms that break a ``triviality barrier'' of $2^n$ has often resulted both in problem-specific structural insight, and in general techniques that have seen further applications; these include \emph{branch and reduce}, \emph{measure and conquer}, \emph{split and list}, as well as powerful algebraic techniques; see~\cite{FominK10, Nederlof26} for surveys.

\subparagraph{Further related work.} 
The optimal discretization problem we study in this paper has also been studied before under the names \emph{axis-parallel red-blue separation}~\cite{BonnetGL19}, \emph{minimum linear classification}~\cite{lin_class}, or \emph{supervised discretization}~\cite{YangWebb}, and further variants of both problems have also been considered, e.g., see~\cite{MisraMS20, KujalaE07, Panos, Megiddo88}. 

Both problems can be seen as special cases of \emph{rectangle stabbing}, a problem that has been widely studied in computational geometry; in contrast to optimal discretization and point separation, this more general version is $W[1]$-hard~\cite{dom}. Another generalization of optimal discretization, where lines of arbitrary slopes are allowed, has also been shown to be $W[1]$-hard~\cite{BonnetGL19}.

A higher dimensional variant of point separation has also been defined, and C\u{a}linescu et al.~\cite{CalinescuDKW05} show that in $d$ dimensions the problem admits a $d$-approximation.

\subparagraph{Overview of our approach.} Our algorithmic idea is very simple, and is based on the naive search algorithm described earlier. We distinguish between lines in the solution that are essentially forced, and lines that are ``free''. We show that, once free lines of one orientation are known, the remaining lines in the solution (of both orientations) can be inferred. Moreover, we argue that there exists an optimal solution for the discretization problem that contains at most $0.4n$ free lines of one orientation.  
For point separation we show the stronger bound of $n/3$. 
We note that when the input point set is in general position (all $x$- and $y$-coordinates distinct), these upper bounds hold, in fact, for the number of all lines of one orientation, i.e., for $\min\{|H|,|V|\}$.

These structural results allow reducing the search space, yielding the stated running times. We show the bound to be optimal for point separation; for optimal discretization it is not far from optimal in the following sense: there is a two-color point set of arbitrarily large size in which \emph{every} optimal solution must have at least $0.3n$ free lines of both orientations. Closing the gap between these upper and lower bounds is an intriguing question we leave open.

\section{Algorithm for optimal discretization}\label{sec2}
The input to the optimal discretization problem consists of two sets $R,B \subset \mathbb{R}^2$ of points in the plane where $|R| + |B| = n$. For a point $p$ we denote its two coordinates by $p.x$ and $p.y$.

The task is to find two sets of lines $H,V \subset \mathbb{R}$, that together separate input points of different classes as follows. For every pair $r \in R$ and $b \in B$ there must be some $v \in V$ such that $r.x < v < b.x$ or $b.x < v < r.x$ holds, or some $h \in H$ such that $r.y < h < b.y$ or $b.y < h < r.y$ holds. We also refer to the horizontal and vertical lines in $H$ and $V$ (given by one coordinate each) as \emph{separators}. In the optimization version of the problem $|H| + |V|$ should be \emph{minimized}, i.e., we would like to separate the points in the sets $R$ and $B$ with as few separators as possible. In the decision version of the problem $k = |V| + |H|$ is given and the question is whether one can separate the points within this budget.  

Note that the values in $V$ and $H$ are disjoint from the $x$-, resp., $y$- coordinates of the input points, i.e., the separating lines avoid the input points. 

\subparagraph{Baseline algorithm.} We now describe the simple search algorithm that also underlies our improved approach.  

\begin{algorithm}[ht!]
\caption{Greedy algorithm}
\label{alg:greedy-algorithm}
\KwIn{A set $P=R\cup B$ of points, a set $H$ of horizontal lines.}
\KwOut{A set $V'$ of vertical lines.}

Let $p_1, \dots, p_{n}$ be the points in $P$ sorted non-decreasingly by $x$-coordinate and let $\varepsilon$ be the smallest absolute difference between two different $x$-coordinates. 

\(V'\leftarrow \emptyset\)\;

\For{\(i=2\) \KwTo \(n\)}{
    \For{\(j=i-1\) \KwTo \(1\)}{
    \If{$(p_j.x < p_{i}.x)$ $\land$ $(p_i$ and $p_j$ are from different sets $R,B)$ $\land$ $(p_i$ and $p_j$ are not separated by $H \cup V')$}{
        $V' \leftarrow V' \cup \{p_i.x - \varepsilon/2\}$\;

        \textbf{break}\;
    }
    }
}

\Return $V'$;
\end{algorithm}

Let $H$ and $V$ be the separators in an optimal solution. Suppose $H$ is given. Then an optimal set $V'$ with $|V'| = |V|$ can be found in polynomial time. The approach is described as Algorithm~\ref{alg:greedy-algorithm}, with its correctness shown in Lemma~\ref{lem:greedy}. The case when $V$ is given is clearly symmetric, and a corresponding optimal set $H'$ can be found similarly.

\begin{lemma}\label{lem:greedy}
If $H$ is a set of horizontal separators, such that $(H,V)$ is a feasible set of separators for an input $R,B$, then Algorithm~\ref{alg:greedy-algorithm} returns a set $V'$ of vertical separators such that $(H,V')$ is feasible and $|V'| \leq |V|$. (Note that $V$ is unknown to the algorithm.)
\end{lemma}
\begin{proof}
We call a pair $(r,b)$ of points with $r \in R$ and $b\in B$ a \emph{colorful pair}. 
Every colorful pair $(r,b)$ must be separated by a line in $H$ or a line in $V$. Consider any solution $(H,V)$. Then, for all colorful pairs not separated by $H$, there must be a separator $v \in V$ with $r.x < v < b.x$ or $b.x < v < r.x$. This means that $V$ is a hitting set for the collection of open intervals defined by the $x$-coordinates of colorful pairs not already separated by $H$. 

Interval hitting set is well-known to be solvable in polynomial time by a greedy algorithm: visiting the intervals in the order of their right endpoints, we hit each (not yet hit) interval at its rightmost point. (Or, in case of open intervals, by a point in the interval, to the right of any interval endpoints inside it.)

\cref{alg:greedy-algorithm} implements exactly this greedy algorithm, therefore its output $V'$ is of optimal size and feasible together with the set of horizontal separators $H$.
Viewing the task as interval hitting set also suggests a more efficient implementation ($\fO(n \log{n})$ instead of $\fO(n^2)$ time), by first sorting and collecting all essential intervals (not containing other intervals) by careful data structuring. As polynomial factors are not our concern here, we omit this optimization.  
\end{proof}

The overall baseline algorithm is now straightforward.
We guess whether $H$ or $V$ is smaller in an optimal solution, breaking ties arbitrarily. Observe that $\min\{|H|,|V|\} \leq n/2$, since a solution of size $n-1$ can always be found.

Suppose, without loss of generality that $|H| \leq |V|$. Then, setting $\upalpha = 1/2$, we iterate over all essentially different placements $H$ of up to $\upalpha n$ horizontal lines, and call Algorithm~\ref{alg:greedy-algorithm} with each choice of $H$.  If $H$ is part of an optimal solution, then Algorithm~\ref{alg:greedy-algorithm} finds a matching $V'$ so that $H$ and $V'$ form an optimal solution. Some choices of $H$ may not work when points share coordinates, since no feasible complementary set $V$ may exist. The algorithm rejects such choices. It still succeeds because the horizontal part of an optimal solution is among the choices considered. Finally, we return a feasible pair $(H, V')$ of smallest total size. 
The algorithm is given as Algorithm~\ref{alg:baseline-algorithm}. Note that we placed no restriction on whether the points can share $x$- or $y$- coordinates. 

\begin{algorithm}[h!]
\caption{Baseline algorithm}
\label{alg:baseline-algorithm}
\KwIn{A set $P=R\cup B$ of points.}
\KwOut{An optimal set $H,V$ of horizontal and vertical lines that separate $P$.}

Assume $|H| \leq |V|$ (other case symmetric). 

Let $y_1, \dots, y_{n'}$ be the distinct $y$-coordinates of points in $P$ in non-decreasing order and let $\varepsilon$ be the smallest absolute difference between two different $y$-coordinates. 

$\upalpha \leftarrow 1/2$\;

\For{\(i=0\) \KwTo \(\lceil \upalpha n \rceil \)}{
    \For{all subsets $H$ of size $i$ of $\{y_2 - \varepsilon/2, \dots, y_{n'} - \varepsilon/2\}$}{ 
        \textbf{call} Algorithm~\ref{alg:greedy-algorithm} to find a set $V_H$ forming a solution together with $H$\;
    }    }
\Return solution $(H, V_H)$ that minimizes $|H| + |V_H|$;
\end{algorithm}

The runtime is dominated by the loop over all sets $H$ of size at most $\upalpha n$. Here, the possible values for horizontal separators fall between neighboring values in the sorted order of distinct $y$-coordinates, their total number is thus at most $n-1$.
The total running time is thus of the form $n^{\fO(1)} \cdot \sum_{i=0}^{\upalpha n} {n-1 \choose i} \subseteq \fO^{*}(2^n)$.

In the following subsection we describe our improved algorithm; it closely resembles the above baseline, with a structural observation that allows reducing the search space. In \S\,\ref{sec3} we adapt the approach to the point separation problem.

\subsection{A faster algorithm}\label{faster-algorithm}
\label{sec2a}

Our improved algorithm is similar to the one given earlier in the section (Algorithm~\ref{alg:baseline-algorithm}). The improvement comes from setting $\alpha = 0.4$, with a resulting reduction in the search space. The subroutine Algorithm~\ref{alg:greedy-algorithm} remains unchanged. Additional care is needed with separators ``forced'' by pairs of input points with equal $x$- or $y$-coordinates.


For any feasible solution $(H,V)$, we say that a separator $x \in V$ or $x \in H$ \emph{uniquely separates} a colorful pair if, after removal of $x$ from the solution, the pair is no longer separated. We take $(H^\prime, V^\prime)$ to be an optimal solution minimizing $\min\{|H'|,|V'|\}$. Without loss of generality, we assume that $|H^\prime| \leq |V^\prime|$.  We start with the following observation.

\begin{observation}\label{obs-separate}
Every line in $H^\prime\cup V^\prime$ uniquely separates at least one colorful pair. 
\end{observation}

\begin{proof}
If a line does not uniquely separate any colorful pair, then it can be removed from the solution, contradicting the optimality of $(H^\prime, V^\prime)$.
\end{proof}

 We now distinguish two types of lines. We call \emph{free lines} those that uniquely separate at least two colorful pairs, and we call all other lines \emph{pinned lines}. Let $H^\prime_f$ and $H^\prime_p$ denote the set of free, respectively, pinned lines in $H^\prime$. Let $V^\prime_f$ and $V^\prime_p$ denote the set of free, respectively, pinned lines in $V^\prime$.

\begin{lemma}\label{free-line-upper-bound}
    $|H^\prime_f|\le 0.4n$.
\end{lemma}

In the following, we say that the vertical lines in $V^\prime$ divide the plane into vertical \emph{strips}, defined as the region between two neighboring vertical lines, with two strips of unbounded width to the left and to the right. Each strip is split by the horizontal lines in $H^\prime$ into \emph{cells} (the part of the strip between two neighboring horizontal lines, as well as a topmost and a bottommost cell of unbounded height). We call a cell \emph{empty} if it contains no point of $R \cup B$ and \emph{nonempty} otherwise.

\begin{proof}[Proof of Lemma~\ref{free-line-upper-bound}]
    We prove the claim through a simple charging scheme. We assign one unit of charge to every point. We charge the lines of $H'$ according to the following rule.
    \begin{itemize}
    \item For every line $h\in H_f^\prime$, if it is a boundary between two consecutive nonempty cells, every point in these two cells gives $0.5$ unit of charge to $h$. Otherwise, $h$
        receives no charge.
    \end{itemize}

Since each line $h\in H_f^\prime$  uniquely separates at least two colorful pairs, $h$ receives at least $1.5$ units of charge (as two colorful pairs require at least $3$ points). Thus, $H_f^\prime$
receives at least $1.5|H_f'|$ units of charge. Observe that every strip contains at least one point. According to the above rule, every strip leaves at least one unit of charge unassigned (as the topmost and bottommost points in the strip each have $0.5$ unit of charge left unassigned). Since there are $|V^\prime|+1$ strips, at least $|V^\prime|+1$ units of charge are unassigned. Then, we get
\begin{center}
\(\displaystyle
n ~\ge~ 1.5|H_f'| + |V^\prime|+1.
\)
\end{center}
According to the assumption that $|H^\prime| \leq |V^\prime|$, we have
\begin{center}
\(\displaystyle
\begin{aligned}
n 
& ~ \ge ~ 1.5|H_f'| + |H^\prime|+1 \\
& ~ \ge ~ 2.5|H_f'| +1.
\end{aligned}
\)
\end{center}
This proves ~\cref{free-line-upper-bound}.
\end{proof} 

We next show that the optimal solution can be assumed to be in a certain canonical position. Let  $\varepsilon$ be the smallest absolute difference between two different $y$-coordinates of the input points. For each $h\in H'$, let $f(h)$ be the smallest $y$-coordinate of an input point strictly above $h$. We modify $H'$ as follows. 
\begin{itemize}
    \item  For each $h\in H'$, move $h$ to $f(h)-\varepsilon/2$.
    \item If a horizontal line $h\in H'$ can be moved to $y'-\varepsilon/2$ while preserving the feasibility of $(H', V')$, where $y'>f(h)$ is a $y$-coordinate of an input point, move $h$ to $y'-\varepsilon/2$. Repeat this step until no horizontal line can be moved.
\end{itemize}

In the first step, the movement of $h$ does not cross any input point, and hence the feasibility of $(H', V')$ is preserved. Thus, $(H', V')$ is still an optimal solution. 
Note that $(H', V')$ is still an optimal solution minimizing $|H'|$ with $|H^\prime| \leq |V^\prime|$. 

In the following, we assume that $(H',V')$ is an optimal solution that is processed in the above way, and $H_f', V_f', H_p', V_p'$ are its free and pinned lines, respectively. 
Given $H_f'$, we claim that $H'$, can be recovered in polynomial time.

\begin{lemma}\label{recover-pinned-lines}
    Let $(H',V')$ be an optimal solution chosen as above. If $H_f'$ is known, then $H'$ can be recovered in polynomial time.
\end{lemma}

\begin{proof}
   
Initially, let $H''=H_f'$.
We process all colorful pairs $(r,b)$ with $r.x=b.x$ in non-decreasing order of $\max(r.y, b.y)$.
If a pair is not yet separated by the current set $H''$, let $H'' = H'' \cup \{\max\{r.y,b.y\}-\varepsilon/2\}$. This procedure runs in polynomial time.

We now prove that $H''=H'$. Let $H_p'=\{h_1,\ldots,h_{|H_p'|}\}$, where $h_1<\cdots<h_{|H_p'|}$.
By the definition of pinned lines, each $h_i$ uniquely separates exactly one colorful pair $(r_i,b_i)$ with $r_i.x=b_i.x$. Otherwise, $h_i$ could be replaced by a vertical
line, contradicting that $(H', V')$ is an optimal solution minimizing $|H'|$.

Let $u_i=\max\{r_i.y,b_i.y\}$. We claim that $u_i=f(h_i)$. Otherwise, since $u_i>f(h_i)$ (if $u_i<f(h_i)$, $h_i$ can not separate $(r_i,b_i)$), we could move $h_i$ to $u_i-\varepsilon/2$
while preserving the feasibility of $(H', V')$. This contradicts the choice of $(H',V')$. Thus, $u_1,\ldots, u_{|H_p'|}$ is in increasing order.

We prove by induction that $H''$ adds $h_1,\ldots,h_{|H_p'|}$ in this order. Suppose $H''$ has already added $h_1,\ldots,h_{i-1}$. Among all colorful pairs $(r,b)$ with $r.x=b.x$ not yet separated by $H''$, suppose $(r', b')$ is the colorful pair minimizing its larger $y$-coordinate. We claim that $\max\{r'.y,b'.y\}=u_{i}$. Otherwise, since $(r_{i}, b_{i})$ is not yet separated by $H'$, we get $\max\{r'.y,b'.y\}<u_{i}$ and hence $h_{i}>\max\{r'.y, b'.y\}$. Therefore, $(r', b')$ can not be separated by $H_p'$, contradicting the feasibility of $(H',V')$. Thus, $\max\{r'.y,b'.y\}=u_{i}$ and $h_{i}$ is the next line added to $H''$. Note that the base case $i=1$ is also covered by the above argument, where $H''=H_f'$.

After adding all $|H_p'|$ lines, all colorful pairs $(r,b)$ with $r.x=b.x$ are separated. Consequently, $H''$ adds exactly $h_1,\ldots,h_{|H_p'|}$ in this order. Hence, $H''=H_f'\cup H_p'=H'$.

This proves ~\cref{recover-pinned-lines}.
\end{proof}

After recovering $H'$, we invoke \cref{alg:greedy-algorithm} to obtain a set $V_{H'}$ forming a solution together with $H'$.
By ~\cref{lem:greedy}, we have $|V_{H'}|\le |V'|$. Since $(H',V')$ is optimal, $(H'',V_{H'})$ is also optimal.
The overall algorithm is described as~\cref{alg:improved-baseline-algorithm}.

\begin{algorithm}[h!]
\caption{Improved algorithm}
\label{alg:improved-baseline-algorithm}
\KwIn{A set $P=R\cup B$ of points.}
\KwOut{An optimal set $H,V$ of horizontal and vertical lines that separate $P$.}

Assume $|H| \leq |V|$ (other case symmetric). 

Let $p_1, \dots, p_{n}$ be the points in $P$ sorted non-decreasingly by $y$-coordinate. Let $y_1, \dots, y_{n'}$ be the distinct $y$-coordinates of points in $P$ in non-decreasing order and let $\varepsilon$ be the smallest absolute difference between two different $y$-coordinates. 

$\upalpha \leftarrow 0.4$\;

\For{\(i=0\) \KwTo \(\lfloor \upalpha n \rfloor \)}{
$\triangleright$ guess $i$ free lines of $H$ and recover pinned lines in $H$

    \For{all subsets $H$ of size $i$ of $\{y_2 - \varepsilon/2, \dots, y_{n'} - \varepsilon/2\}$}{

        \For{\(j=2\) \KwTo \(n\)}{
    \For{\(k=j-1\) \KwTo \(1\)}{
    \If{$(p_j.x = p_{k}.x)$ $\land$ $(p_k.y < p_{j}.y)$ $\land$ $(p_j$ and $p_k$ are from different sets $R,B)$ $\land$ $(p_j$ and $p_k$ are not separated by $H)$}{
        $H \leftarrow H \cup \{p_j.y - \varepsilon/2\}$\;

        \textbf{break}\;
    }
    }
}
        
        \textbf{call} Algorithm~\ref{alg:greedy-algorithm} to find a set $V_H$ forming a solution together with $H$\;
    }    }
\Return solution $(H, V_H)$ that minimizes $|H| + |V_H|$;
\end{algorithm}

\subparagraph{Analysis of the running time.} In the following, we use the standard estimate on the binomial sum $\sum_{i=0}^{\lfloor \upbeta n \rfloor}{{n \choose i} \leq 2^{n H(\upbeta)}}$, for $0 < \upbeta \leq 1/2$, where $H(x) = -x \log_2{x} - (1-x) \log_2{(1-x)}$ is the binary entropy function (e.g., see~\cite[\S\,3.2]{FominK10}).

As in the proof of Lemma~\ref{recover-pinned-lines}, recovering pinned lines of $H$ takes polynomial time. Since \cref{alg:greedy-algorithm} takes polynomial time, the total running time is dominated by the number of candidate sets of free lines of $H$, which is of the form $\sum_{i=0}^{0.4 n} {n \choose i}$. Thus, the the total running time is $n^{\fO(1)} \cdot \sum_{i=0}^{0.4 n} {n \choose i} \subseteq \fO(1.9602^n)$, where we used the estimate for the binomial sums with $H(0.4) \leq 0.970952$.


\subparagraph{Lower bound.} An improvement to the running time of our algorithm would follow simply from lowering the value $\upalpha=0.4$ in Lemma~\ref{free-line-upper-bound}. While we cannot rule out such an improvement, we show that the value is already close to optimal. 

\begin{lemma}\label{lower-bound}
    There exists an infinite family of instances with $n$ points such that every optimal solution $(H,V)$ satisfies $\min\{|H_f|,|V_f|\}\ge 0.3n$.
\end{lemma}

\begin{figure}[htbp]
    \centering 
    \includegraphics[width=0.7\textwidth]{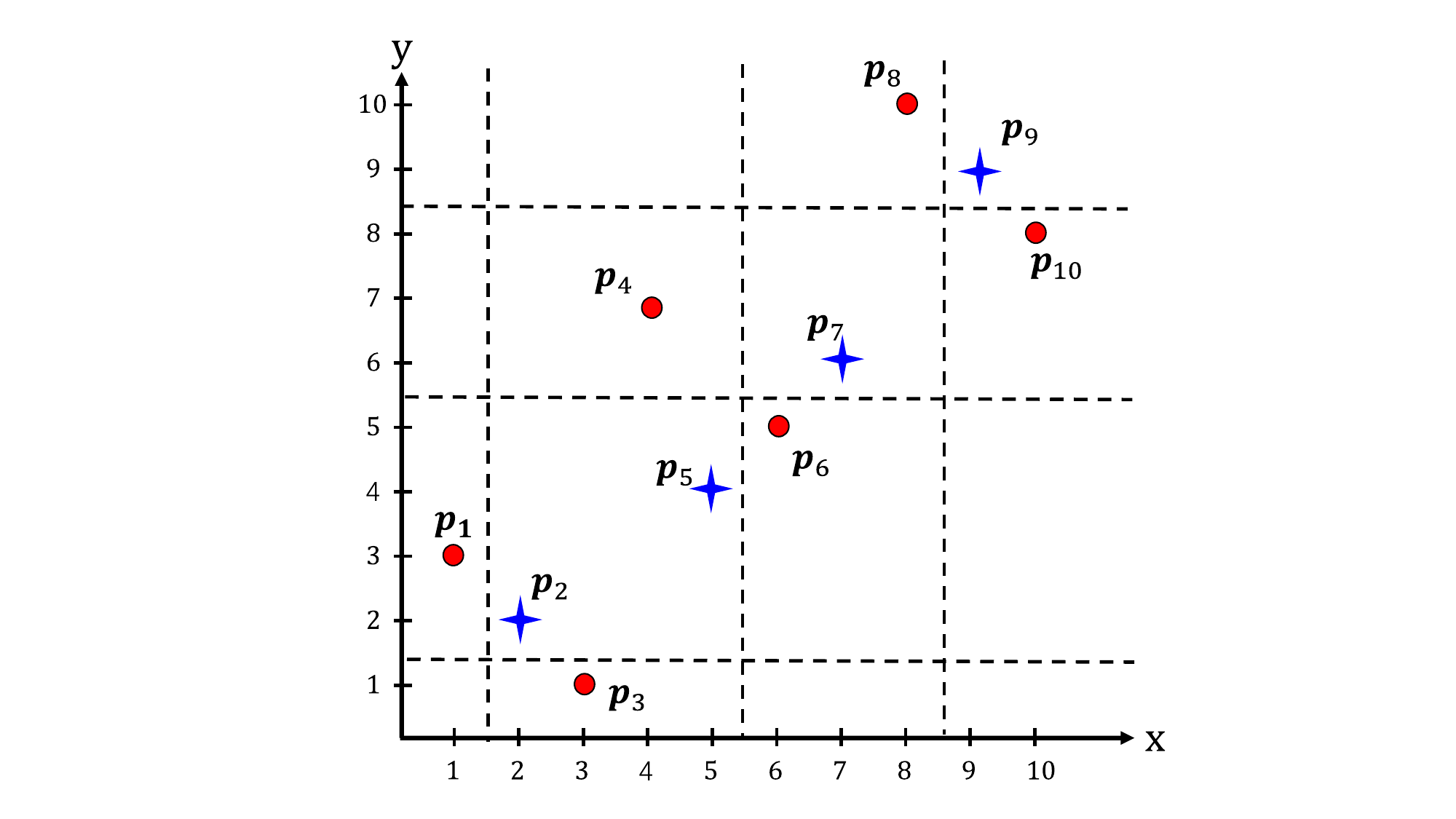} 
    \caption{Base instance with $10$ points; points in $R$ shown as dots, points in $B$ as stars. Dashed lines represent a feasible solution.} 
    \label{base-instance} 
\end{figure}

\begin{proof}
We use a base instance shown in \cref{base-instance}, consisting of $10$ points, with $R = \{p_1 = (1,3), p_3 = (3,1), p_4 = (4,7), p_6 = (6,5), p_8 = (8,10), p_{10} = (10,8)\}$, and $B = \{p_2 = (2,2), p_5 = (5,4), p_7 = (7,6), p_{9} = (9,9)\}$. 
The instance has a feasible solution with $V = \{1.5, 5.5, 8.5\}$ and $H = \{1.5, 5.5, 8.5\}$. 

Consider the three triplets: $(p_1,p_2,p_3), (p_5,p_6,p_7), (p_8,p_9,p_{10})$. In the first triplet, separating $p_1$ from $p_2$ requires one line and separating $p_2$ from $p_3$ requires one line. Similarly, the second and third triplet require $2$ lines, and the lines for different triplets are necessarily disjoint; we thus require at least $6$ lines. 

Next, we show that the optimal solution must balance the number of vertical and horizontal lines. In the first triplet, if we choose two vertical lines, then $p_3$ and $p_5$ are not separated. If we choose two horizontal lines, then $p_1$ and $p_5$ are not separated. Similarly for the third triplet, choosing only vertical or only horizontal lines will leave the pairs $p_7$ and $p_8$, resp., $p_7$ and $p_{10}$ not separated. For the second triplet, using two vertical or two horizontal lines would leave $p_4$ unseparated from $p_5$ or $p_7$.

Therefore each triplet requires exactly one vertical line and one horizontal line in every optimal solution. Hence the feasible solution $(H,V)$ is optimal and every optimal solution $(H,V)$ satisfies $|H| = |V| = 3$. 

Since every optimal solution $(H,V)$ satisfies $|H|=|V|=3$ and every colorful pair has distinct $x$- and $y$-coordinates, every line $h\in H$ must uniquely separate at least two colorful pairs. Otherwise, if $h$ uniquely separates exactly one colorful pair, then $h$ can be replaced by a vertical line, yielding an optimal solution with $|H|=2$ and $|V|=4$, a contradiction. As is shown in \cref{base-instance}, the three horizontal lines, from bottom to top, uniquely separate two, three, and two colorful pairs, respectively. Hence $H=H_f$. Symmetrically, $V=V_f$.

For any integer $t \ge 1$, we can construct an instance with $n = 10t$ points by copying the base instance $t$ times and placing the copies along the diagonal, shifting both coordinates of the $i$-th copy by $10(i-1)$. The previous argument applies to each copy of the base instance, implying that an optimal solution must contain $3t$ horizontal and $3t$ vertical lines, so 
$\min\{|H|,|V|\} = 3t = 0.3n$ and $\min\{|H_f|,|V_f|\} = 3t = 0.3n$. 
\end{proof}

\subparagraph{Note on general position.} The above algorithms do not assume general position of the input, i.e., the results hold even if input points share $x$- or $y$-coordinates. If the input is in general position, then we can assume that the smaller of the two sets $H$ and $V$ has only free lines. Indeed, suppose $|H| \leq |V|$. Then every pinned horizontal line (that separates exactly one colorful pair) could be replaced by a vertical line separating the same pair (since the two points are not aligned), preserving the optimality of the solution. Thus, Lemma~\ref{free-line-upper-bound} can be understood to imply $\min\{|H|,|V|\} \leq 0.4n$ for all general-position inputs. 

\section{Algorithm for point separation}\label{sec3}

The input to the point separation problem is a set $P \subset \mathbb{R}^2$ of points in the plane, where $|P| = n$. The task is to find two sets of lines $H,V \subset \mathbb{R}$, that together separate input points as follows. For every pair $r,b \in P$ there must be some $v \in V$ such that $r.x < v < b.x$ or $b.x < v < r.x$ holds, or some $h \in H$ such that $r.y < h < b.y$ or $b.y < h < r.y$ holds.

We reuse terminology from \S\,\ref{sec2} and observe that Algorithm~\ref{alg:greedy-algorithm} and Algorithm~\ref{alg:baseline-algorithm} can be adapted to point separation with minimal modifications, yielding a runtime of $\fO^*(2^n)$ or $\fO^*(n^{k/2})$.

Our improved result (Theorem~\ref{thm2}) again follows from the observation that the number of free lines of one type can be assumed to be small. (In this uncolored setting, free lines are those that separate at least two pairs.) In this case, however, we can show our bound to be tight, through an infinite family of instances.

\begin{lemma}\label{upper-bound-n-color}
    For every instance $P$ of size $n$ for point separation, there exists an optimal solution $(H,V)$ with $\min\{|H_f|,|V_f|\}\le \frac{n-1}{3}$.
\end{lemma}

Lemma~\ref{upper-bound-n-color} implies Theorem~\ref{thm2} by setting $\alpha = 1/3$ and using the earlier estimate for the binomial sum with $2^{H(1/3)} \leq 1.8899$.

\begin{proof}[Proof of Lemma~\ref{upper-bound-n-color}]
We follow a similar argument as in Lemma~\ref{free-line-upper-bound}, taking $(H^\prime, V^\prime)$ to be an optimal solution minimizing $\min\{|H'|,|V'|\}$, assuming without loss of generality that $|H'| \leq |V'|$. 
Vertical strips, and cells are defined as before. 

We now understand ``colorful pairs'' as all pairs of points that must be separated. With this remark, 
\cref{obs-separate} holds unchanged and the remainder of the proof simplifies. 

We adapt the charging scheme of the proof of~\cref{free-line-upper-bound} as follows. We assign one unit of charge to every point, and charge the lines of $H'$ according to the following rule.
\begin{itemize}
\item For every line $h\in H_f^\prime$, if it is a boundary between two consecutive nonempty cells, every point in these two cells gives $0.5$ unit of charge to $h$. Otherwise, $h$
        receives no charge.
\end{itemize}

Since each line $h\in H_f^\prime$  uniquely separates at least two colorful pairs, and the points participating in these pairs have to be in separate cells, as well as in separate strips, at least $4$ points are involved, and $h$ receives at least $2$ units of charge. Thus, $H_f^\prime$
receives at least $2|H_f'|$ units of charge. Observe that every strip contains at least one point. According to the above rule, every strip leaves at least one unit of charge unassigned (as the topmost and bottommost points in the strip each have $0.5$ unit of charge left unassigned). Since there are $|V^\prime|+1$ strips, at least $|V^\prime|+1$ units of charge are unassigned. Then, we get
\begin{center}
\(\displaystyle
n ~\ge~ 2|H_f'| + |V^\prime|+1.
\)
\end{center}
According to the assumption that $|H^\prime| \leq |V^\prime|$, we have
\begin{center}
\(\displaystyle
\begin{aligned}
n 
& ~ \ge ~ 2|H_f'| + |H^\prime|+1 \\
& ~ \ge ~ 3|H_f'| +1.
\end{aligned}
\)
\end{center}
This proves ~\cref{upper-bound-n-color}.
\end{proof}

Let  $\varepsilon$ be the smallest absolute difference between two different $y$-coordinates of the input points. For each $h\in H'$, let $f(h)$ be the smallest $y$-coordinate of an input point strictly above $h$. We again assume that $H'$ is brought into a canonical state as follows. 
\begin{itemize}
    \item  For each $h\in H'$, move $h$ to $f(h)-\varepsilon/2$.
    \item If a horizontal line $h\in H'$ can be moved to $y'-\varepsilon/2$ while preserving the feasibility of $(H', V')$, where $y'>f(h)$ is a $y$-coordinate of an input point, move $h$ to $y'-\varepsilon/2$. Repeat this step until no horizontal line can be moved.
\end{itemize}

As shown in \S\,\ref{faster-algorithm}, such movements preserve the feasibility and optimality of $(H', V')$. Note that $(H', V')$ is still an optimal solution minimizing $|H'|$ with $|H^\prime| \leq |V^\prime|$.  
Given $H_f'$, we claim that $H'$ can be recovered in polynomial time.

\begin{lemma}\label{recover-pinned-lines-point-separation}
    For every instance $P$ of size $n$ for point separation, and an optimal solution $(H',V')$ as chosen above, if $H_f'$ is known, then $H'$ can be recovered in polynomial time.
\end{lemma}

\begin{proof}
Understanding ``colorful pairs'' as all pairs of points that must be separated, the proof of \cref{recover-pinned-lines} also applies here.
Hence $H'$ can be recovered in polynomial time.
\end{proof}

Thus, after setting $\alpha=1/3$ and understanding ``colorful pairs'' in the way above,  ~\cref{alg:improved-baseline-algorithm} solves the point separation problem.

\subparagraph{Analysis of the running time.} 
As in the proof of Lemma~\ref{recover-pinned-lines-point-separation}, recovering pinned lines of $H$ takes polynomial time. Since \cref{alg:greedy-algorithm} takes polynomial time, the total running time is dominated by the number of candidate sets of free lines of $H$, which is of the form $\sum_{i=0}^{n/3} {n \choose i}$. Thus, the the total running time is $n^{\fO(1)} \cdot \sum_{i=0}^{n/3} {n \choose i} \subseteq \fO(1.8899^n)$, where we used the earlier estimate for the binomial sums with $H(n/3) \leq 0.918296$.

Next, we show that the value $\alpha = 1/3$ is optimal.
\begin{lemma}\label{lower-bound-n-color}
    For every $n=3k+1$ 
    there is a point separation instance with $n$ points such that every optimal solution $(H,V)$ satisfies $\min\{|H_f|,|V_f|\}\ge \frac{n-1}{3}$.
\end{lemma}

\begin{figure}[htbp]

    \centering 
    \includegraphics[width=0.5\textwidth]{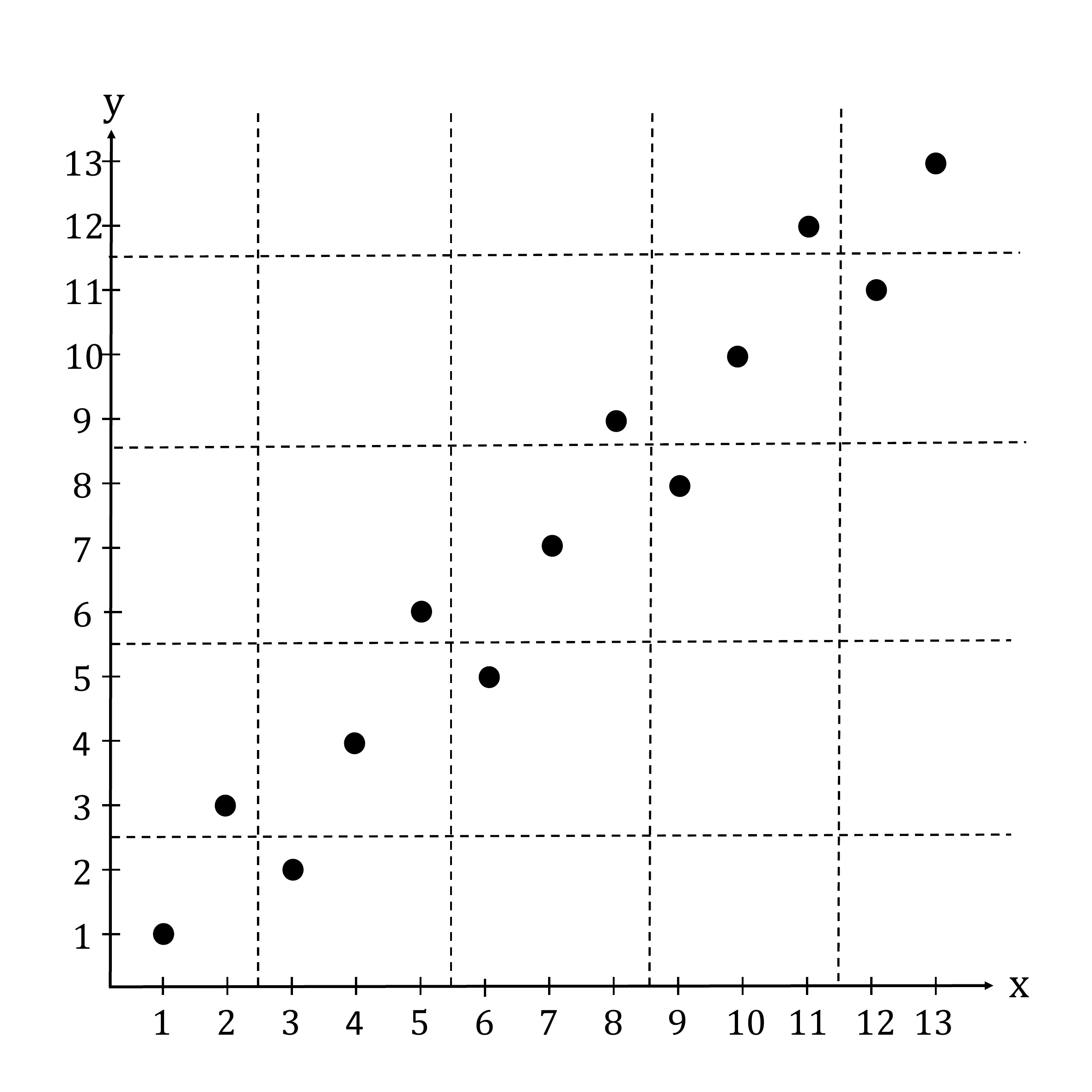} 
    \caption{An instance with $k=4$; points shown as dots. Dashed lines represent a feasible solution.} 
    \label{n-color-instance} 
\end{figure}

\begin{proof}[Proof of Lemma~\ref{lower-bound-n-color}]
For any integer $k\ge 1$, we construct an instance with $n = 3k+1$ points, as shown in \cref{n-color-instance}. We start from a base instance with $3$ points: $\{(1,1), (2,3), (3,2)\}$, placed in $k$ shifted copies along the diagonal with the $i$-th copy defined as $\{(3i-2,3i-2), (3i-1,3i), (3i,3i-1)\}$. Finally, place
one extra point $(3k+1,3k+1)$. To separate the points we need at least $2$ lines per copy. If we choose two horizontal lines, then the point $(3i-1, 3i)$
and the first point of the next copy (or the final point) are not separated. Similarly, if we choose two vertical lines, then the point $(3i, 3i-1)$
and the first point of the next copy (or the final point) are not separated. In these two cases, one more line is required between every two adjacent
copies. Therefore, every optimal solution requires one horizontal line and one vertical line per copy, and hence $|H| = |V| = k =\frac{n-1}{3}.$ 

Since every optimal solution $(H,V)$ satisfies $|H|=|V|$ and every colorful pair has distinct $x$- and $y$-coordinates, every line $h\in H$ must uniquely separate at least two colorful pairs. Otherwise, if $h$ uniquely separates exactly one colorful pair, then $h$ can be replaced by a vertical line, yielding an optimal solution with $|H| < |V|$, a contradiction. As is shown in \cref{n-color-instance}, every horizontal line uniquely separates exactly two colorful pairs. Hence $H=H_f$. Symmetrically, $V=V_f$. Hence, $\min\{|H_f|,|V_f|\}= \frac{n-1}{3}$.
\end{proof}
The remark from the end of \S\,\ref{sec2} also applies here. Thus in case of general position input, the bound applies to the quantity $\min\{|H|,|V|\}$ of an optimal solution $(H,V)$.

\section{Conclusions}

The most interesting structural question left open by our work is the gap between $0.3n$ and $0.4n$ for the optimal cardinality of the smaller of the horizontal and vertical separator-sets in the two-colored case; an optimal bound of $n/3$, matching the monochromatic case would seem plausible.

Further improving the running times of our algorithms, possibly by different techniques than those used here, remains an interesting direction; note that the possibility of an algorithm with runtime of the form $2^{o(n)}$ has not been ruled out.

A different question of a flavor similar to those we studied is to separate points with horizontal and vertical line \emph{segments} that form a \emph{rectangulation}. This question would likely require different techniques, and is meaningful both under minimization of the number of line segments and minimization of the total line length. 

\subparagraph{Disclosure about use of AI.} The authors used ChatGPT 5.5 and Claude Opus 4.7 both to explore the problem structure, and to verify arguments. The writeup is the authors' own and they assume full responsibility for the contents.

\newpage

\bibliography{lipics-v2021-sample-article}

\end{document}